\documentclass[a4paper, 12pt,oneside]{article}

\usepackage{amsmath,amsfonts,amssymb,bbm}
\usepackage{graphicx}
\usepackage{multirow}

\usepackage{amsthm}
\newtheorem{definition}{Definition}[section]
\newtheorem{lemma}{Lemma}[section]
\newtheorem{theorem}{Theorem}[section]

\usepackage{anysize}
\marginsize{2.1cm}{2.1cm}{1.8cm}{1.8cm}

\usepackage{float,framed}
\floatstyle{ruled}

\usepackage{caption}

\newcommand{\doubleroot}[1]{
	{\lceil{2\sqrt{#1}-1}\rceil}
}

\newcommand{\interior}{\operatorname{int}}

\newcommand{\Mod}{~\mathrm{mod}~}

\begin{document}

\title{Counting Blanks in Polygonal Arrangements}

\author{
Arseniy Akopyan%
\thanks{Institute of Science and Technology Austria (IST Austria), Am Campus~1, 3400 Klosterneuburg, Austria
{akopjan@gmail.com}
}
\and
Erel Segal-Halevi%
\thanks{Ariel University (Ariel 40700) and Bar-Ilan University (Ramat-Gan 52900), Israel.
{erelsgl@gmail.com}
}}

\maketitle

\begin{abstract}
Inside a two dimensional region (``cake''), there are $m$ non-overlapping tiles of a certain kind (``toppings'').
We want to expand the toppings while keeping them non-overlapping, and possibly add some blank pieces of the same ``certain kind'', such that the entire cake is covered. How many blanks must we add?

We study this question in several cases: (1) The cake and toppings are general polygons. (2) The cake and toppings are convex figures. (3) The cake and toppings are axes-parallel rectangles. (4) The cake is an axes-parallel rectilinear polygon and the toppings are axes-parallel rectangles.
In all four cases, we provide tight bounds on the number of blanks.
\end{abstract}

%

\section*{Introduction}
Consider a two-dimensional cake $C$ with $m$ non-overlapping
toppings $Z_1,Z_2,\ldots,Z_m$. We
want to cut the cake without harming the toppings.
I.e, we want to partition the entire cake to non-overlapping pieces $Z_1',Z_2',\ldots,$ such that each topping is contained in a piece and each piece contains at most a single topping. 
There are some geometric constraints on the pieces, e.g, they should be polygonal or convex or rectangular. 
This might require us to add some ``blanks'' --- pieces with no topping. For example, in the rectangular cake at the left of Figure \ref{fig:intro} there are $m=4$ rectangular toppings. If the pieces must be axes-parallel rectangles, then we will need to add at least one blank, denoted at the right by $Z_5'$.
\begin{figure}
\begin{center}
\includegraphics[width=0.3\textwidth]{graphics/rect-13.mps}
\hskip .15\textwidth
\includegraphics[width=0.3\textwidth]{graphics/rect-14.mps}
\end{center}
\caption{
\label{fig:intro}
\textbf{Left:} a rectangular cake with 4 rectangular toppings.
\\
\textbf{Right:} an expansion of the toppings to 4 larger pieces and a fifth blank piece.
}
\end{figure}
Given $m$ and the geometric constraint on the pieces, how many blanks must we add in the worst case (for the worst initial arrangement of toppings)?

Besides cutting cakes, an additional application of this question is for re-division of land \cite{segalhalevi2018redividing}. 
There are some small lots on a piece of land. The owners have built on their initial lots and need to keep them. They would like to expand their lots and then fill the land with additional lots --- but not too many additional lots. All lots are required to have a ``nice'' geometric shape. How may additional lots do they need to add in the worst case?

This paper answers this question under three different geometric constraints on the toppings and pieces: polygonality, convexity and rectangularity. 
In all these cases, we prove that for any initial arrangement of toppings there exists a \emph{maximal expansion} --- where the toppings expand inside the cake until they cannot expand any further while keeping the geometric constraints.
Since we are interested in a worst-case upper bound, it is sufficient to consider such maximal arrangements, since the initial arrangement might already be maximal. So our question becomes: how many blanks can there be in a maximal arrangement of pieces? We answer this question in four cases.

\textbf{1.} The cake and toppings are polygons. 
Then in any maximal arrangement, the entire cake is covered --- there are no blanks.

\textbf{2.} The cake and toppings are convex figures. 
Then in any maximal arrangement, the uncovered spaces are all convex, and there are at most $2m-5$ of them.

\textbf{3.} The cake and toppings are  axes-parallel rectangles, as in Figure \ref{fig:intro}.
Then in any maximal arrangement,
all uncovered spaces are axes-parallel rectangles, and there are at most $m-\doubleroot{m}$ of them

\textbf{4.} The toppings are still axes-parallel rectangles, but the cake may be an arbitrary simply-connected axes-parallel polygon. Let $T$ be the number of reflex vertices (270-degree interior angles) in the cake. 
Then in any maximal arrangement, the remaining uncovered spaces can be partitioned into $b$ axes-parallel rectangles, where $b\leq m-\doubleroot{m}+T$.

All the results are tight in the following sense: In each case, for every $m$ (and $T$), there is an explicit construction where the number of blanks equals the bound.  

In addition, for cases 2 and 3 we consider a related question. Suppose the toppings lie in the unbounded plane $\mathbb{R}^2$, and we do not expand them.
We define a \emph{hole} as a connected component of the plane that remains outside of toppings (a hole can be unbounded). How many holes can there be?
We prove the following answers.

\textbf{2'} In any arrangement of $m$ pairwise-disjoint convex figures in the plane, the number of holes is at most $2m-4$.

\textbf{3'} In any arrangement of $m$ pairwise-disjoint axes-parallel rectangles in the plane, the number of holes is at most $m-2$.

Both these bounds are tight in the same sense as above. It is interesting that the numbers are equal up to one in case 2, but quite different in case 3.

For each $k\in\{1,2,3,4\}$, case $k$ is handled in Section $k$.

In each case, the proof consists of two steps. First, we show that any initial arrangement of toppings has a maximal expansion. Then, we prove that in any maximal arrangement, the number of blanks is bounded as claimed. 
This method works in cases 1--4, but it may not work in other cases; some counter-examples are discussed in Section \ref{sec:conclusion}.

\section{Polygonal Cake and Polygonal Pieces}
In this warm-up section, the cake and each of the $m$ toppings is a polygon. Initially, we allow polygons to be not simply-connected, in contrast to the classic definition.
We prove that in this case it is possible to expand the toppings such that there will be no blanks.
\begin{theorem}
\label{thm:conn}
Let $C$ (``cake'') be a polygon and $Z_1, \dots,Z_m$ (``toppings'') be pairwise-disjoint polygons contained in $C$. Then there exists a partition of $C$ into polygons, $C=Z'_1 \sqcup \cdots \sqcup Z'_{m}$,\footnote{The symbol $\sqcup$ denotes union of pairwise-interior-disjoint sets.}
where $Z_i\subseteq Z'_i$ for all $i\leq m$.
\end{theorem}
\begin{proof}
The new partition can be created by the following procedure.

Let $C^*$ be the cake outside the toppings, $C^* := C\setminus \bigcup_{i = 1}^m Z_i$. Since $C$ and the $Z_i$ are all polygons, $C^*$ is a union of a finite number of polygons, say: $C^* = H_1\sqcup \cdots \sqcup H_n$, where the sides of each $H_j$ are made of subsets of the sides of $C$ and the sides of the toppings. Moreover, every $H_j$ must have at least one side that overlaps with a side of some topping $Z_i$ (since the cake itself is connected). So $Z_i\cup H_j$ is a polygon. Replace $Z_i$ with $Z_i \cup H_j$ and repeat the procedure for the remaining components of $C^*$.
\end{proof}
Note: Theorem \ref{thm:conn} implies that, if the initial arrangement of toppings is \emph{any} maximal arrangement, then it has no blanks.

While Theorem \ref{thm:conn} is easy to prove, it is not so easy to extend to toppings that are connected but not polygonal. A counter-example is given in Section \ref{sec:conclusion}.

On the other hand, Theorem \ref{thm:conn} can be extended to the case where the cake and the toppings are \emph{simply-connected polygons}. Initially, 
add to each topping thin polygonal tubes that connect it to the cake boundary and to each of the other toppings, as long is it is possible to do so without overlapping other toppings (see Figure \ref{fig:tubes}). 
Then proceed as in Theorem \ref{thm:conn}: add each component of $C^*$ to an adjacent topping. The tubes guarantee that the toppings remain simply-connected.

\begin{figure}
\begin{center}
\includegraphics[angle=90]{graphics/rect-15.mps}
\hskip .1\textwidth
\includegraphics[angle=90]{graphics/rect-16.mps}
\end{center}
\caption{
\label{fig:tubes}
\textbf{Left:} a cake with three simply-connected topping.
\\
\textbf{Right:} Tubes connecting the toppings to each other and to the cake boundary.
}
\end{figure}

\section{Convex Cake and Convex Pieces}
In this section, the cake and each of the $m$ toppings is a \emph{convex figure}. 

\begin{theorem}
\label{thm:convex case}
Let $C$ (``cake'') be a convex figure and $Z_1, \dots,Z_m$ (``toppings'') be pairwise-disjoint convex figures in $C$, where $m\geq 3$.
There exists a partition of $C$ into $m+b$ convex figures, $C=Z'_1 \sqcup \cdots \sqcup Z'_{m+b}$,
where $Z_i\subseteq Z'_i$ for all $i\leq m$, and $b\leq 2 m - 5$.
Moreover, for every $m$ there exists an arrangement of $m$ toppings where in every such partition, $b = 2 m - 5$.
\end{theorem}

The techniques and example in the proof are very similar to \cite{Edelsbrunner1990Covering}.
We first prove the second part of the theorem by showing an arrangement for which $b = 2m-5$.

\subsection{Lower bound}
\begin{figure}
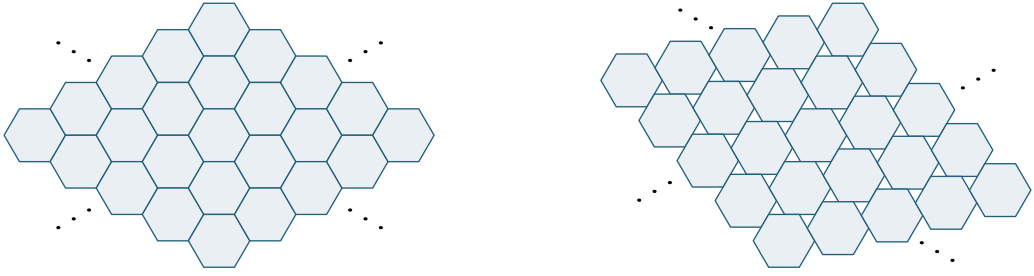

	\begin{center}
		\includegraphics{graphics/fig-top-1.mps}
		\hskip 2cm
		\includegraphics{graphics/fig-top-2.mps}
	\end{center}
	\caption{\label{fig:holes-convex-hexagons} 
		\textbf{Left:} a tiling of the plane by hexagons.
		\\
		\textbf{Right:} a modification of the tiling where near each vertex there is a blank.
	}
\end{figure}
To get intuition for the lower bound construction, consider the tiling of the plane with hexagons, shown in Figure \ref{fig:holes-convex-hexagons}/Left.
Each point in which $3$ hexagons meet is called a \emph{vertex}.
By slightly moving the hexagons, it is possible to create, near each vertex, a triangular blank, as is shown in Figure \ref{fig:holes-convex-hexagons}/Right.
No hexagon can be expanded towards an adjacent blank while remaining convex and disjoint from the other hexagons (i.e, the arrangement is maximal). Each blank touches three hexagons and each hexagon touches six blanks. Hence, the number of blanks is asymptotically twice the number of hexagons, which gives an initial approximation $b\approx 2 m$.
In a finite tiling, there are boundary conditions. For example, when the tiling is contained in a square cake, there are $\Theta(\sqrt{m})$ triangles near the boundary of $C$. These can be discarded or joined with nearby toppings, so the total number of blanks is only $2m-\Theta(\sqrt{m})$. 

Figure \ref{fig:holes-convex-triangles} shows a more sophisticated arrangement of toppings that has only a constant number of blanks in the boundary. The arrangement is based on a construction of Edelsbrunner, Robison and Shen \cite{Edelsbrunner1990Covering}. For every integer $k\geq 0$, the toppings are:
\begin{itemize}
\item $3$ large external hexagons;
\item $3k$ medium intermediate hexagons;
\item $3$ or $4$ or $5$ small inner polygons, depending on whether $m \Mod 3$ is $0$ or $1$ or $2$.
\end{itemize}
Near the vertices of the toppings, there are triangular blanks. Note that the arrangement is \emph{maximal} --- there is no way to expand any topping into any blank. 
Therefore the final partition consists exactly of the $m$ toppings and $b$ blanks shown in the figure.
We now calculate $b$. Each blank is triangular so it is adjacent to $3$ toppings, so below each blank is counted three times:
\begin{itemize}
\item Each external hexagon is adjacent to $3$ blanks, for a total of $9$;
\item Each intermediate hexagon is adjacent to $6$ blanks, for a total of $18k$;
\item Regarding the inner polygons:
\begin{itemize}
\item  When $m \Mod 3=0$ --- there are three inner quadrangles, which are adjacent to $3\cdot 4=12$ blanks;
\item  When $m \Mod 3=1$ --- there are one triangle and three pentagons, which are adjacent to $3+3\cdot 5=18$ blanks;
\item  When $m \Mod 3=2$ --- there are a triangle, a quadrangle, a pentagon and 2 hexagons, which are adjacent to $3+4+5+2\cdot 6=24$ blanks.
\end{itemize}
\end{itemize}
All in all, the number of toppings and blanks is one of the following $\forall k\geq 0$  (see Figure \ref{fig:holes-convex-triangles}/Top):
\begin{itemize}
\item $m=6+3k$ and $b=(9+18k+12)/3=7+6k=2m-5$;
\item $m=7+3k$ and $b=(9+18k+18)/3=9+6k=2m-5$;
\item $m=8+3k$ and $b=(9+18k+24)/3=11+6k=2m-5$.
\end{itemize}
or one of the following (see Figure \ref{fig:holes-convex-triangles}/Bottom):
\begin{itemize}
	\item $m=3$ and $b=1=2m-5$ 
	\item $m=4$ and $b=3=2m-5$
	\item $m=5$ and $b=5=2m-5$
\end{itemize}
In all cases the number of blanks is $2m-5$, as claimed.\qed

\begin{figure}
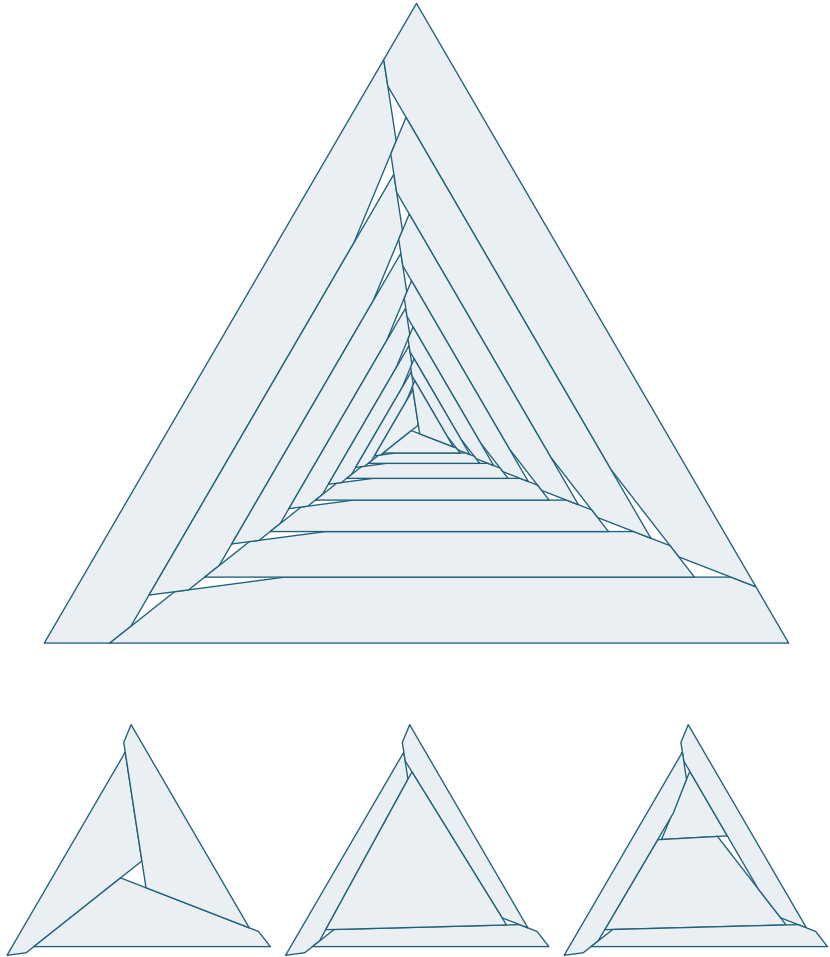

	\begin{center}
		\includegraphics{graphics/alteredel-1.mps}
		\vskip 1cm
		\includegraphics{graphics/alteredel-2.mps}
		\includegraphics{graphics/alteredel-3.mps}
		\includegraphics{graphics/alteredel-4.mps}
	\end{center}
	\caption{\label{fig:holes-convex-triangles} 
		\small
		\textbf{Top:} A maximal arrangement of hexagonal toppings in a triangle cake.
		\\
		\textbf{Bottom:} 3 or 4 or 5 toppings for the inner part of the cake, depending on $m \Mod 3$.
	}
\end{figure}

\subsection{Upper bound}
For the upper bound we first prove a theorem for an unbounded cake.
\begin{theorem}
\label{thm:nonconvex holes}
Let $Z_1,\ldots,Z_m$, with $m\geq 3$, be pairwise-interior-disjoint convex figures in the plane. Define a \emph{hole} as a connected component of $\mathbb{R}^2\setminus \cup_{i=1}^m Z_i$. Then there are at most $2m-4$ holes. This bound is tight.
\end{theorem}
\begin{proof}
Define a graph $G(V,E)$ where the 
vertices are the $m$ toppings and 
there is an edge between each two toppings whose boundaries meet. Since the toppings are convex, every two toppings meet at most once (in a single point or segment), so $G$ is planar and simple. 
Each hole is a face of $G$. Therefore, it is sufficient to prove that $|F|\leq 2|V|-4$, where $F$ is the set of faces of the graph.

We can assume w.l.o.g. that $G$ is connected; otherwise we can just add edges between connected components of $G$, since this does not change $V$ or $F$.
Therefore, by Euler's formula:  $|V|-|E|+|F|=2$.
Every edge touches at most 2 faces. Every interior face touches at least 3 edges. In a connected planar graph with at least 3 vertices, the exterior face too touches at least 3 edges. Therefore: $2 |E| \geq 3|F|$. Substituting in Euler's formula gives the result.

Tightness is proved by the same example of the previous subsection, adding 1 for the unbounded region outside the cake.
 \end{proof}
We now return to proving the upper bound in Theorem \ref{thm:convex case}. First, we expand the toppings inside the cake $C$ by the following procedure. Pick an arbitrary topping, say $Z_1$. 
Among all convex figures in $C$ containing $Z_1$ and not overlapping any other topping, choose one that is inclusion-maximal.\footnote{The existence of such maximal element can be proved based on the Kuratowski--Zorn lemma. The proof is straightforward and we omit it. See Lemma 3.3 in \cite{Mohammadi2017Reconstruction}.} Replace $Z_1$ with this maximal element. Repeat the procedure for the remaining toppings.

This procedure results in a \emph{maximal arrangement} --- an arrangement where no topping can be expanded further while remaining convex and disjoint from the others. Pinchasi \cite[Claim 2]{Pinchasi2015Perimeter} proved that, in an maximal arrangement, every hole inside the cake is a convex polygon and does not have a common boundary with the cake.\footnote{
An alternative proof of this fact follows from the proof for the case of rectangle cake and rectangle pieces in Subsection  \ref{sub:rect-upper-a}. It is proved there that every hole must be inner and have only convex vertices; the same proof is valid for the case of convex cake and convex pieces.
}

By Theorem \ref{thm:nonconvex holes}, the total number of holes inside the cake, plus the 
entire unbounded region outside the cake, is at most $2m-4$. Therefore, the total number of bounded holes inside the cake is at most $2m-5$. Since each hole is convex, we can make each hole a single blank piece in the final partition, and get  $b\leq 2m-5$ as claimed.  \qed

\section{Rectangular Cake and Rectangular Pieces}
\label{sec:rect}
In this section, the cake and each of the $m$ toppings is an axes-parallel rectangle.
\begin{theorem}
\label{thm:boxes}
Let $C$ (``cake'') be an axes-parallel rectangle and $Z_1, \dots,Z_m$ (``toppings'') be pairwise-disjoint axes-parallel rectangles in $C$.
There exists a partition of $C$ into $m+b$ axes-parallel rectangles, $C=Z'_1 \sqcup \cdots \sqcup Z'_{m+b}$,
where $Z_i\subseteq Z'_i$ for all $i\leq m$, and $b\leq m-\doubleroot{m}$.
Moreover, for every $m$ there exists an arrangement where in all partitions $b = m-\doubleroot{m}$.
\end{theorem}

In the proofs below we will use the following property of the function $\lceil{2\sqrt{m}-1}\rceil$:
\footnote{We are grateful to an anonymous referee for suggesting the form $\lceil{2\sqrt{m}-1}\rceil$ for this function.}
\[
\lceil{2\sqrt{m}-1}\rceil= 
\begin{cases} 
	2k & 	k^2<m\leq k^2+k\\
	2k+1& 	k^2+k<m\leq (k+1)^2
\end{cases}.
\]

\subsection{Lower bound}\label{sub:rect-lower}
\begin{figure}
\begin{center}
\includegraphics{graphics/fig-top-3.mps} \hskip 1cm
\includegraphics{graphics/fig-top-4.mps} \hskip 1cm
\includegraphics[width=4cm]{graphics/rect-11.mps}
\end{center}
\caption{\label{fig:holes-rect-worst} 
\textbf{Left:} a tiling of the plane by squares.
\\
\textbf{Middle:} a modification of the tiling where near each vertex there is a blank.
\\
\textbf{Right:} an arrangement contained in a rectangular cake. There are $m=16$ toppings and $m-\doubleroot{m}=9$ blanks.
}
\end{figure}
Consider the tiling in Figure \ref{fig:holes-rect-worst}/Left. Here, four squares meet near each vertex. By moving them slightly, we get the arrangement in Figure \ref{fig:holes-rect-worst}/Middle, where near each vertex there is a square blank. Each blank touches four squares and each square touches four blanks. Hence, the number of squares and blanks is asymptotically the same.

When the tiling is finite, there are boundary conditions. Suppose the cake is a square and it is tiled by $m=(k+1)^2$ smaller squares in a grid of $k+1$ times $k+1$. Then, the number of inner vertices, that can be converted to blanks, is $k^2 = m - (2k + 1) = m-\doubleroot{m}$. See Figure \ref{fig:holes-rect-worst}/Right.

When $m$ is not a square number, e.g. $m=k^2+t$ for some $t>0$, the lower bound can be constructed by gluing $k$ squares at the right and then $k+1$ squares at the top of a $k\times k$ tiling. Every glued square adds a $4$-vertex that can be converted to a blank, except the first one at the right and the first one at the top. These are exactly the points where the function $\doubleroot{m}$ increases by 1. Therefore the number of blanks always remains $m-\doubleroot{m}$. \qed

\subsection{Upper bound part A: All holes are internal axes-parallel rectangles}\label{sub:rect-upper-a}
We first expand the toppings as follows. For the topping $Z_1$ consider a rectangle $Z_1'$ with maximal area among all axis-parallel rectangles contained in $C$, containing $Z_1$ and avoiding all other toppings. Substitute $Z_1$ by $Z_1'$. Repeat for the remaining toppings.
\begin{figure}
\begin{center}
	\includegraphics{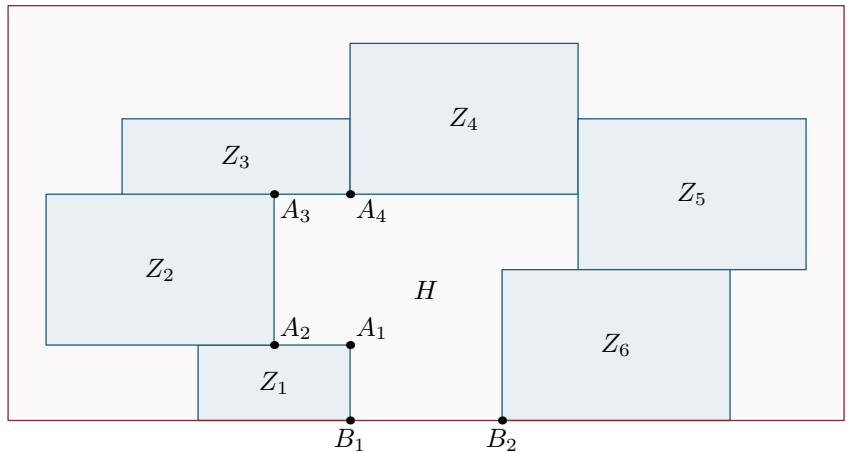}
\caption{\label{fig:holes-rect}A rectangular cake $C$, rectangular toppings $Z_j$, and a hole $H$. Note the configuration is not maximal.}
\end{center}
\end{figure}

We now have a maximal arrangement of toppings.
Again we define a \emph{hole} as a connected component of $C\setminus \cup_{i=1}^m Z_i$.
It is clear that in a maximal arrangement every hole is simply-connected, since a hole that is not simply-connected contains a topping which can be expanded. We will now prove that in any maximal arrangement, every hole is an axes-parallel rectangle and is not adjacent to the cake boundary. 
We need several definitions regarding a hole $H$ (the examples refer to Figure \ref{fig:holes-rect}):
\begin{definition}
\label{def:vertex}
A \emph{hole-vertex} of $H$ is a point on the boundary of $H$ that is a vertex of a topping or of the cake. A hole-vertex is called:
\begin{itemize}
\item \emph{convex} --- if the internal angle adjacent to it is less than $180^\circ$ (i.e, $90^\circ$; like $A_2,A_3, B_1,B_2$);
\item \emph{nonconvex} --- if the internal angle adjacent to it is at least $180^\circ$ (like. $A_4, A_1$).
\end{itemize}
\end{definition}

\begin{definition}
\label{def:edge}
A \emph{hole-edge} of $H$ is a line-segment between adjacent hole-vertices of $H$. A hole-edge is called:
\begin{itemize}
\item \emph{inner-edge} --- if it is contained in the interior of $C$;
\item \emph{boundary-edge} --- if it is contained in the boundary of $C$.
\end{itemize}
\end{definition}

\begin{definition}
A hole-vertex is called:
\begin{itemize}
\item \emph{inner-vertex} --- if it is contained in the interior of $C$ (links two inner-edges; like $A_i$);
\item \emph{boundary-vertex} --- if it is
on the boundary of $C$ (touches a boundary-edge; like $B_i$).
\end{itemize}
\end{definition}
\begin{definition}
A hole is called:
\begin{itemize}
\item \emph{inner-hole} --- if it is contained in the interior of $C$ (so it has only inner-edges);
\item \emph{boundary-hole} --- if it has a common boundary with $C$ (so it has some boundary-edges).
\end{itemize}
\end{definition}

Any inner-edge $e$ of a hole $H$ represents an opportunity to expand some topping into $H$. Such opportunity can only be blocked by an adjacent topping with an edge $e'$ perpendicular to $e$. Let $v$ be the vertex that connects $e$ and $e'$. We call $v$ a \textbf{blocking-vertex} of $e$. For example, $A_3$ is a blocking-vertex of the edge $A_3 A_4$. It is easy to see that, if $v$ is a blocking-vertex of $e$:
\begin{itemize}
\item $v$ is a convex vertex --- it connects two perpendicular edges of different toppings.
\item $v$ is an inner vertex --- it connects  two edges of toppings.
\item $v$ blocks only $e$ --- it cannot simultaneously block $e'$.
\end{itemize}
In a maximal arrangement, each inner-edge must have at least one blocking-vertex. Therefore:
\begin{align}
\label{eq:every-hole}
\text{
For every hole $H$:
\#inner-convex-vertices$(H)$ $\geq$ \#inner-edges$(H)$
}
\end{align}
We now consider inner-holes and boundary-holes separately.

An \emph{inner hole} has only inner-vertices and inner-edges, and their number must be equal, so:
\begin{align}
\label{eq:inner-hole}
\text{
  For every inner-hole $H$:~~~
 \#inner-edges$(H)$ $=$ \#inner-vertices$(H)$
}
\end{align}
Combining (\ref{eq:every-hole}) and (\ref{eq:inner-hole}) implies that, in an inner-hole, all vertices are convex, so:
\begin{align}
\label{eq:inner-hole-rectangle}
\text{
Every inner-hole is a rectangle.
}
\end{align}

A \emph{boundary-hole}'s boundary contains sequences of adjacent boundary-edges and  sequences of adjacent inner-edges. Each boundary sequence contains an alternating sequence of boundary-vertex --- boundary-edge --- boundary-vertex --- boundary-edge --- ... boundary-vertex.  So:
\begin{align}
\label{eq:boundary-hole}
\text{
For every boundary-hole $H$:~~~
\#boundary-vertices$(H)$ $\geq$
\#boundary-edges$(H) + 1$
}
\end{align}
Summing (\ref{eq:every-hole}) and (\ref{eq:boundary-hole}) gives:
\begin{align*}
\text{
For every boundary-hole $H$:~~~
\#vertices$(H)$ $\geq$ \#edges$(H) + 1$
}
\end{align*}
But this is impossible, since in every hole the total number of vertices and edges must be equal. Therefore:
\begin{align}
\label{eq:no-boundary-hole}
\text{
There are no boundary holes.
}
\end{align}
(\ref{eq:inner-hole-rectangle}) and  (\ref{eq:no-boundary-hole}) imply that, in a maximal arrangement, all holes look like Figure \ref{fig:holes-rect-worst}: each hole is a rectangle adjacent to four toppings.

\subsection{Upper bound part B: Counting the holes} \label{sub:rect-upper-b}
Our upper bound matches the lower bound of Subsection \ref{sub:rect-lower}: $m-\doubleroot{m}$.
The bound relies on the following lemmas:

\begin{lemma}
\label{lem:k1k2}
For all positive integers $k_1$, $k_2$, $t$:
\begin{equation*}
	(k_1-1)(k_2-1) - t 
	\leq 
	k_1 k_2 - t - \doubleroot{k_1 k_2 - t}.
\end{equation*}
\end{lemma}
\begin{proof}
First, we prove the inequality for $t=0$.
It is sufficient to show that $\doubleroot{k_1 k_2}\leq k_1+k_2-1$.
Indeed, for a fixed sum $k_1+k_2$, the product $k_1 k_2$ is maximized when $k_1=k_2$ or when $k_1+1=k_2$.
In the former case  $\doubleroot{k_1 k_2}=2k_1-1$, in the latter case $\doubleroot{k_1 k_2}=\lceil 2 \sqrt{k_1^2+k_1}\rceil - 1 = (2 k_1+1)-1 = 2k_1$. In both cases 
the inequality holds with equality.

For $t>0$, the difference (RHS minus LHS) is even larger, so the inequality remains true.
\end{proof}

Suppose a rectangle is partitioned into $m$ smaller rectangles (with no holes). Define a \emph{$3$-vertex} as a point in which three rectangles meet and a \emph{$4$-vertex} as a point in which four rectangles meet. 
\begin{lemma}
\label{lem:toy model}
Suppose a rectangular cake $C$ is partitioned into $m$ rectangles. Then, the number of $4$-vertices is at most $m-\doubleroot{m}$.
\end{lemma}
\begin{figure}
	\begin{center}
		\includegraphics{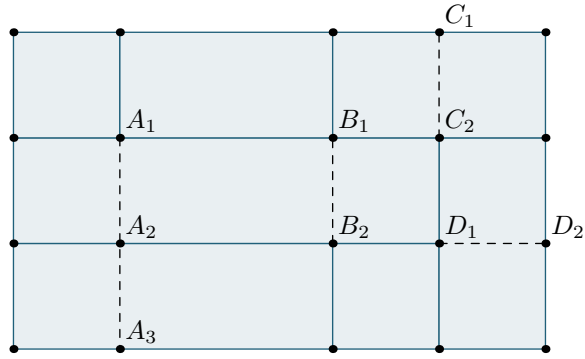}
	\end{center}
	\caption{\label{fig:toy-rect} 
		Bounding the number of $4$-vertices in a rectangle tiling.
	}
\end{figure}
\begin{proof}
First, suppose that there are only $4$-vertices (no $3$-vertices). Then $C$ is partitioned, by the lines containing all sides of all rectangles, to a grid of $k_1$ by $k_2$ smaller rectangles. So the total number of $4$-vertices is $(k_1-1)(k_2-1)$ and the total number of rectangles is $m=k_1 k_2$. Then the lemma follows from Lemma \ref{lem:k1k2} setting $t=0$.

Next, suppose that there are 3-vertices (refer to Figure \ref{fig:toy-rect} for the examples).
Choose one $3$-vertex (e.g. $A_1$).
Add a segment that separates the rectangle that does not have vertex at $v$ to two parts (e.g. the segment $A_1 A_2$).
This converts the $3$-vertex to $4$-vertex.
Additionally, there are several cases depending on where the other end of the additional segment lands:

(a) If it lands on the boundary of $C$ (like with $D_1 D_2$ and $C_2 C_1$), then no further action is required; the number of rectangles grows by $1$ and the number of $4$-vertices grows by $1$.

(b) If it lands on the boundary of another rectangle (like with $A_1 A_2$), then a new $3$-vertex is created, and can be handled in the same way by continuing the segment. This must stop because eventually the line hits the boundary of $C$.

(c) If it lands on a $3$-vertex (like with $B_1 B_2$), then an additional $4$-vertex is created; the number of rectangles grows by $1$ and the number of $4$-vertices grows by $2$.

In all cases, the number of $4$-vertices grows at least as much as the number of rectangles.

Continue this procedure until all $3$-vertices disappear. Then, the partition is a grid formed by the lines containing the sides of the all original rectangles. Suppose the grid has $k_1\times k_2$ rectangles (e.g. in Figure \ref{fig:toy-rect}, $k_1=3$ and $k_2=4$). So, the total number of rectangles is $k_1 k_2$ and the total number of $4$-vertices is $(k_1-1)(k_2-1)$. 
Suppose we had $t$ removal steps. So, the original number of $4$-vertices was at most $(k_1-1)(k_2-1)-t$ and the original number of rectangles $m$ was exactly $k_1 k_2-t$. From Lemma \ref{lem:k1k2} it follows that the number of original $4$-vertices is at most $m- \doubleroot{m}$.
\end{proof}

\begin{figure}
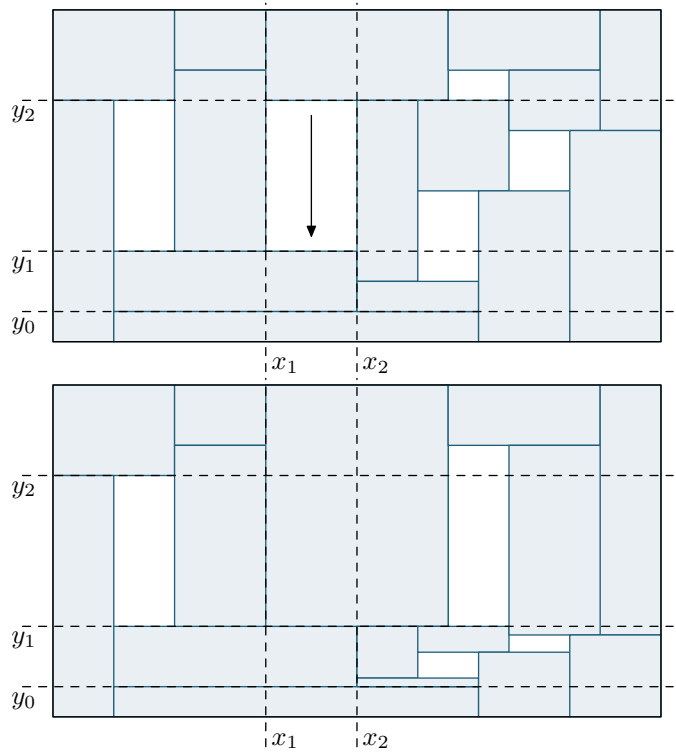

	\begin{center}
\includegraphics		{graphics/rect-1.mps}\\
\includegraphics		{graphics/rect-2.mps}
	\end{center}
	\caption{\label{fig:contracting_holes} 
\textbf{Top:} An arrangement with $5$ holes.
\textbf{Bottom:} One of the holes is contracted vertically. The other rectangles and holes change their size but remain in the picture.
}
\end{figure}
We now return to the proof of Theorem \ref{thm:boxes}.
Recall from Subsection \ref{sub:rect-upper-a} that we have a maximal arrangement of toppings, where each hole is a rectangle bounded by four toppings. Each of these toppings is blocked from expanding into the hole, either by the next topping counter-clockwise 
(as in Figure \ref{fig:holes-rect-worst}) or by the next topping clockwise.
Our goal is to show that the total number of such holes is at most $m-\doubleroot{m}$.
We will remove the holes one by one, by contracting each hole to a 4-vertex without changing the number of toppings. 

Each hole is contracted in two steps: \emph{vertical} and \emph{horizontal}.
The vertical contraction is illustrated in Figure \ref{fig:contracting_holes}. Consider the hole $[x_1,x_2]\times[y_1,y_2]$. There are two cases regarding the toppings surrounding the hole: either the top has a side at $x=x_1$ and the bottom has a side at $x=x_2$ (in the clockwise case), or vice versa (in the counter-clockwise case). These cases are entirely analogous. We assume the former case, as in Figure \ref{fig:contracting_holes}/Top. 

Let $y_0$ be the bottom coordinate of the topping below the hole. Transform the arrangement of rectangles in the following way. Every point $(x,y)$, where $x>x_2$ and $y\in[y_0,y_2]$, is transformed to $(x,y')$, where:
\begin{align*}
y' := y_0 + \frac{y_1-y_0}{y_2-y_0}\cdot (y-y_0),
\end{align*}
so the ray $(x>x_2, y=y_2)$ goes to $(x>x_2, y=y_1)$ and the ray $(x>x_2, y=y_0)$ remains in its place. 
Now, the rectangle on top of the hole can be extended down, so that it covers the hole (see \ref{fig:contracting_holes}/Bottom). 

The horizontal contraction is very similar. Let $x_0$ be the left coordinate of the rectangle to the left of the hole. Every point $(x,y)$, where $y<y_1$ and $x\in[x_0,x_2]$, is transformed to $(x',y)$, where:
\begin{align*}
x' := x_0 + \frac{x_1-x_0}{x_2-x_0}\cdot (x-x_0),
\end{align*}
so the ray $(y<y_1, x=x_2)$ goes to $(y<y_1, x=x_1)$ and the ray $(y<y_1, x=x_0)$ remains in its place. Now, instead of the hole, there is a single $4$-vertex at the point $(x_1,y_1)$. 

The shrinking does not change the combinatorics of the arrangement: the number of toppings does not change, no new holes are created, and no 4-vertices disappear.

After all holes are contracted, the situation is as in Lemma \ref{lem:toy model}, where the number of $4$-vertices is upper-bounded. Therefore, the number of holes in the original configuration is upper-bounded by the same expression. \qed

\subsection{Number of holes without expanding}
For completeness, we formulate the analogue of Theorem~\ref{thm:nonconvex holes} for axis-parallel rectangles.

\begin{theorem}
	\label{thm:nonrectangular holes}
Let $Z_1,\ldots,Z_m$, with $m\geq 3$, be pairwise-interior-disjoint axes-parallel rectangles in the plane. Define a \emph{hole} as a connected component of $\mathbb{R}^2\setminus \cup_{i=1}^m Z_i$. Then there are at most $m-2$ holes. This bound is tight.
\end{theorem}

\begin{proof}
The proof of the upper bound is similar to Theorem~\ref{thm:nonconvex holes}; the only difference is that here, every face of the planar graph touches at least \emph{four} edges (instead of three). Therefore, we have $2|E|\geq 4|F|$ (instead of $2|E|\geq 3|F|$). Substituting this in Euler's formula $|V|-|E|+|F|=2$ gives that $|F|\leq |V|-2$, so the number of holes is at most  $m-2$ as claimed.

An example giving the lower bound is shown on Figure~\ref{fig:many holes}.
Four of the rectangles form a long box and all others are placed between them forming holes.
\end{proof}

\begin{figure}
\begin{center}
\includegraphics		{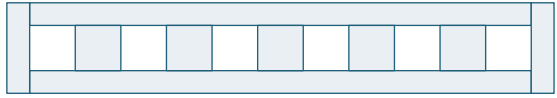}\\
\end{center}
\caption{\label{fig:many holes} 
	An arrangement of $m=9$ rectangles with $m-2=7$ holes (including the outer hole).
}
\end{figure}

\section{Rectilinear-polygonal cake and rectangular pieces}
In this section, the toppings are still axes-parallel rectangles, but now the cake $C$ can be any simply-connected axes-parallel rectilinear polygon. The ``complexity'' of a rectilinear polygon is characterized by the number of its \emph{reflex vertices} --- vertices with internal angle $270^\circ$. It is known that a rectilinear polygon with $T$ reflex vertices can always be partitioned to at most $T+1$ rectangles \cite{Keil2000Polygon,Eppstein2010GraphTheoretic}, and this bound is tight when the vertices of $C$ are in general position. Since our goal is to bound the number of blank rectangles, we expect the bound to depend on $T$, in addition to $m$ (the number of toppings).

\begin{theorem}
	\label{thm:boxes and polygonal cake}
Let $C$ (``cake'') be a simply-connected axes-parallel polygon with $T$ reflex vertices,
and $Z_1, \dots,Z_m$ (``toppings'') be pairwise-disjoint axes-parallel rectangles in $C$.
There exists a partition of $C$ into $m+b$ axes-parallel rectangles, $C=Z'_1 \sqcup \cdots \sqcup Z'_{m+b}$,
where $Z_i\subseteq Z'_i$ for all $i\leq m$, and $b\leq m+T-\doubleroot{m}$. 
Moreover, for every $m$ and $T$ there exists an arrangement where in every such partition, $b = m+T-\doubleroot{m}$.	
\end{theorem}

\subsection{Lower bound}
\begin{figure}
\begin{center}
\includegraphics{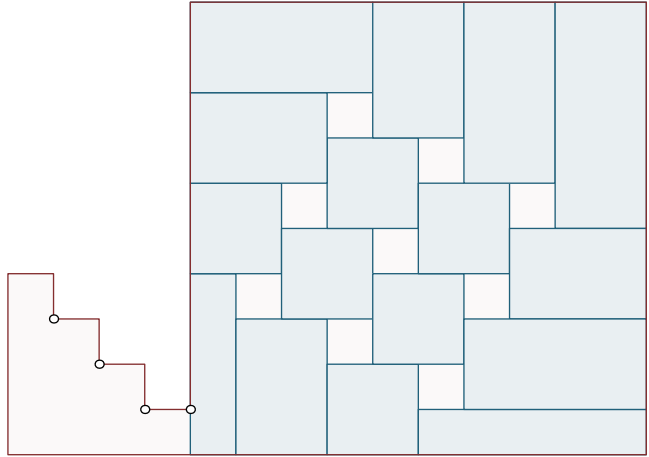}
\end{center}
\caption{\label{fig:holes-rectilinear-worst}
A rectilinear cake made of a union of a rectangle and a rectilinear ``staircase'' with $T=4$ steps. The rectangle part contains a maximal arrangement of $m=16$ toppings and $m-\doubleroot{m}=9$ blanks. The rectilinear part adds $T$ reflex vertices (circled) and $T$ rectangular blanks.
}
\end{figure}
Take a worst-case arrangement of rectangular toppings in a rectangular cake, such as the one shown in Figure \ref{fig:holes-rect-worst}/Right.
For any $T\geq 1$, it is possible to convert the rectangular cake to a rectilinear polygon with $T$ reflex vertices by adding a narrow rectilinear ``staircase'' with $T$ ``steps'', as shown in Figure \ref{fig:holes-rectilinear-worst}.

The number of blanks in the rectangle is $m-\doubleroot{m}$.
The additional rectilinear part is a single \emph{hole}, but it contains $T$ \emph{blanks} (its rectangular components). All in all, we need $m+T-\doubleroot{m}$ blanks.
\subsection{Upper bound}
\begin{figure}
\begin{center}
\includegraphics{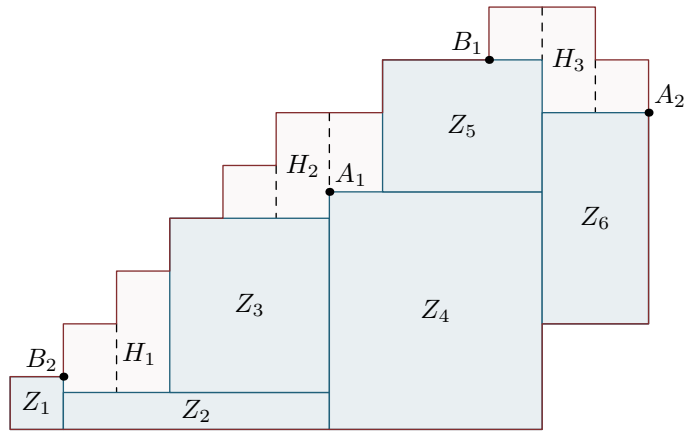}
\end{center}
\caption{\label{fig:holes-rectilinear}A rectilinear cake with six rectangular topping $Z_1, \dots, Z_6$, and three rectilinear holes $H_1$, $\dots$, $H_3$.
Note the configuration is not maximal since $Z_6$ can be extended upwards into $H_3$.
The dashed lines inside the holes indicate a possible partitioning of the holes to eight rectangular blanks.}
\end{figure}
We first expand the arrangement to a maximal arrangement of toppings as in Subsection \ref{sub:rect-upper-a}. 

In the case of a rectangle cake, we proved that all holes are inner-holes and rectangular. This is no longer true here: Figure \ref{fig:holes-rectilinear} shows that we can have boundary-holes that are not rectangular.
It \emph{is} still true that in a maximal arrangement every hole is simply-connected, since a hole that is not simply-connected contains a topping which can be expanded. 

Moreover, it is still true that every inner-edge must be blocked by a hole-edge perpendicular to it, and the blocking vertex must be a convex-vertex. However, now the blocking-vertex can be either an inner-convex-vertex, or a boundary-vertex that is a reflex-vertex of $C$ (like $B_1$ in Figure \ref{fig:holes-rectilinear}). We
call such vertex \emph{connection-reflex-$C$-vertex} since it connects an inner-edge and a boundary-edge. 
So now, for every hole $H$:

\begin{align}
\label{eq:every-hole-2}
\text{
\#connection-reflex-$C$-vertices$(H)$
+
\#inner-convex-vertices$(H)$
 $\geq$
  \#inner-edges$(H)$
}
\end{align}
In an inner-hole there are only inner-vertices. Therefore, conclusion (\ref{eq:inner-hole-rectangle}) is still true --- every inner-hole is a rectangle.

In a boundary-hole, inequality (\ref{eq:boundary-hole}) is still true, since in each sequence of adjacent boundary-edges, there is a vertex for each edge plus one additional vertex:
\begin{align}
\tag{\ref{eq:boundary-hole}}
\text{
For every boundary-hole $H$:
\#boundary-vertices$(H)$ $\geq$
\#boundary-edges$(H) + 1$
}
\end{align}

Summing (\ref{eq:every-hole-2}) and (\ref{eq:boundary-hole}) gives:
\begin{align}
\label{eq:boundary-hole-3}
\text{
\#connection-reflex-$C$-vertices$(H)$
$+$
\#vertices$(H)$
$-$
\#inner-nonconvex-vertices$(H)$
}
\\
\notag
\text{
 $\geq$ \#edges$(H)+1$
}
\end{align}

Now, in every hole, \#edges$(H)=$\#vertices$(H)$. Combining this with (\ref{eq:boundary-hole-3}) gives:
\begin{align}
\label{eq:reflex-vertices}
\text{
\#inner-nonconvex-vertices$(H)$
$+1$
$\leq$
\#connection-reflex-$C$-vertices$(H)$
}
\end{align}
In addition, the boundary of $H$ contains some boundary-reflex-vertices, all of which are reflex-vertices of $C$ that are not connection-vertices:
\begin{align}
\text{
\#boundary-reflex-hole-vertices$(H)$
$=$
\#boundary-reflex-$C$-vertices$(H)$
}
\end{align}
Adding the latter two inequalities gives, for every boundary-hole $H$:
\begin{align}
\label{eq:reflex-vertices-total}
\text{
\#reflex-hole-vertices$(H)$
$+1$
$\leq$
\#reflex-$C$-vertices$(H)$
}
\end{align}

A rectilinear polygon with $x$ reflex vertices can always be partitioned to at most $x+1$ rectangles \cite{Keil2000Polygon,Eppstein2010GraphTheoretic}. Therefore, each boundary-hole $H$ can be partitioned into rectangles whose number is at most the number of reflex-cake-vertices in the boundary of $H$. 
Moreover, every reflex-cake-vertex is a vertex of at most one boundary-hole. Therefore, summing over all boundary-holes gives that all boundary-holes can be partitioned into rectangles whose total number is at most $T$ --- the number of reflex-vertices of $C$. 

Adding the (at most) $m-\doubleroot{m}$ rectangular inner-holes gives that the total number of rectangular blanks is at most: $T+m-\doubleroot{m}$. \qed

\section{Counter-examples and Open Questions}
\label{sec:conclusion}
We considered three families of pieces: polygons, convex figures, and axes-parallel rectangles. For each of these families, we proved that: (a) any arrangement of toppings in the family can be expanded to a maximal arrangement, and (b) the number of blanks in any maximal arrangement is bounded.
These two facts are not necessarily true for other families. We provide several examples below.\footnote{These examples are based on suggestions by anonymous referees.}

\subsection{Rectangles with constrained lengths}
Suppose that all pieces must be axes-parallel rectangles whose lengths are irrational numbers, while the cake is the unit square. Then, even a single topping cannot be expanded to a maximal arrangement. However, it may still be possible to find a partition with a bounded number of blanks. For example, with a single topping it is sufficient to add 3 blanks. We did not calculate how many blanks may be needed in general.
\subsection{Path-connected sets}

\begin{figure}
	\begin{center}
		\includegraphics{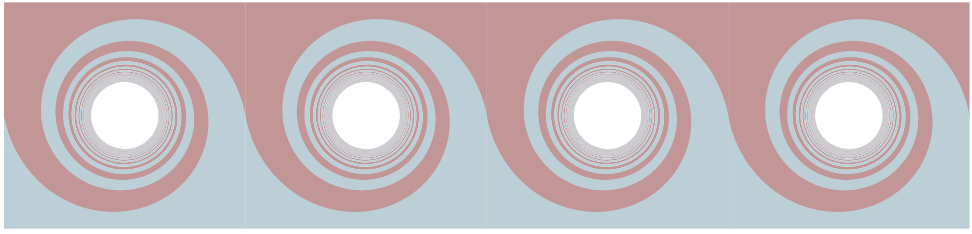}
	\end{center}
		
	\caption{
		\label{fig:spirals} 
A rectangular cake with 
two simply-connected toppings. The circular blanks cannot be attached to any topping without crossing the other one. Therefore there is an unbounded number of blanks (4 of them are shown).
}
\end{figure}

Suppose that all pieces must be path-connected. This very natural constraint is not as simple as it seems: the number of blanks in a maximal arrangement might be unbounded.
An example is shown in Figure~\ref{fig:spirals}. The cake is a rectangle. There are two toppings; they are not only path-connected but also simply-connected.  There are four circular blanks which cannot be attached to any topping, since any path from a blank to a topping must cross the other topping. Therefore the arrangement is maximal.
It is straightforward to extend this example to contain any number of blanks.

If the pieces are required to be connected but not path-connected, then the discs \emph{can} be attached to any of the toppings, so in a maximal expansion of this example there are no blanks. We do not know if this is always the case when the cake and pieces are connected.

\subsection{Closed path-connected sets}
One way to overcome the problem in the previous section is to require that the toppings and pieces be topologically closed, in addition to being path-connected. However, in this case a maximal arrangement might not exist. 

For example, suppose the cake is a union of:
\begin{enumerate}
\item the curve $(x,\sin(1/x))$ for $x\in(0,1]$, 
\item the segment $[(0,-1),(0,1)]$, and --
\item an additional curve connecting the origin and the point $(1, \sin(1))$. 
\end{enumerate}
This cake is closed and path-connected.
The two toppings are any two points on curve (3). They too are closed and path-connected.
But there is no maximal expansion to pieces that are closed and path-connected: if an expanded piece intersects both curve (1) and segment (2) then it is not path-connected, and if a piece contains only curve (1) then it is not closed.

\subsection{Optimality of the greedy algorithm}

\begin{figure}
	\begin{center}
\includegraphics{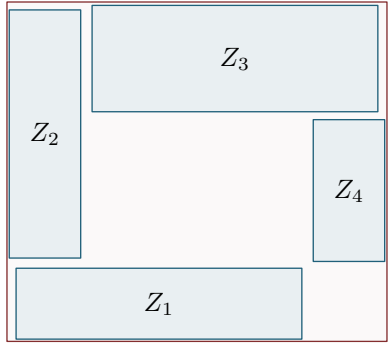}
	\end{center}
\caption{
\label{fig:not_optimal}	
	A configuration of toppings that can be extended in two ways: optimal ($Z_4$ leftwards) and sub-optimal ($Z_4$ downwards).
}
\end{figure}

Suppose we are given a cake with $m$ toppings and we want to expand the toppings such that the number of blanks is as small as possible. Suppose we use the following greedy algorithm:
\begin{framed}
Find a maximal expansion of each topping in turn, in an arbitrary order.
\end{framed}
The results in this paper imply that, in the four cases studied, the number of blanks attained by this greedy algorithm is upper-bounded by some function of $m$. However, our results do \emph{not} imply that the number of blanks is minimal: 
in specific cases it may possible to extend the toppings in a way that attains a smaller number of blanks --- less than the worst-case minimum. For example, in Figure \ref{fig:not_optimal}, if~$Z_4$ is extended downwards, then we get one blank, which is the worst-case minimum for $m=4$ toppings. However, if $Z_4$ is extended leftwards, then we can get zero blanks.

This opens an interesting algorithmic question: given a specific arrangement of toppings, how can we find an expansion with a minimum number of blanks?
A related natural question is how to minimize, instead of the number of blanks, their total perimeter or area.

\subsection{Restricted toppings}
This paper only considered cases where the geometric constraint on the initial toppings is the same as on the final pieces. But when the geometric constraints on the toppings are stricter, the upper bound might be lower.

For example, suppose the toppings must be points. Then, in both the convex case and the rectangular case the upper bound is $b=0$. In the  convex case we can take the Voronoi tesseletion (it works for disks as well; in this case we get weighted Voronoi). In the rectangular case, we can just separate the points by vertical lines, and if two points have same x-coordinate, separate them by horizontal lines.

\section*{Acknowledgements}
This paper started as a question in the Math Overflow website \cite{mathoverflow2016partition}.
Also participated in the discussion:  Steven Stadnicki,  Fedor Petrov and  Mikhail Ostrovskii.
We are grateful to Reuven Cohen, David Peleg, Herbert Edelsbrunner, Rom Pinchasi, Ashkan Mohammadi, Jorge Fernandez and N. Bach for discussions about the problem studied in this paper.

We would like to thank the editorial board of SIDMA and three anonymous referees for the careful reading of our paper and several useful suggestions that significantly improved its readability and correctness.

A. A. is supported by People Programme (Marie Curie Actions) of the European Union's Seventh Framework Programme (FP7/2007-2013) under REA grant agreement n$^\circ$[291734].

E. S.-H. was supported by the Doctoral Fellowships of Excellence at Bar-Ilan University, the Mordecai and Monique Katz Graduate Fellowship Program at Bar-Ilan University, and ISF grant 1083/13.

\bibliographystyle{abbrv}
\bibliography{bib}

\vskip 1cm

\end{document}